\documentclass[draft]{birkjour}

\usepackage[noadjust]{cite}
\usepackage{xcolor}
\RequirePackage[all]{xy}

\usepackage{longtable}
\usepackage{lscape}
\usepackage{afterpage}

\newtheorem{thm}{Theorem}[section]
\newtheorem{cor}[thm]{Corollary}
\newtheorem{lem}[thm]{Lemma}

\theoremstyle{definition}

\theoremstyle{remark}

\newtheorem*{ex}{Example}
\numberwithin{equation}{section}

\newcommand{\BibTeX}{B\kern-0.1emi\kern-0.017emb\kern-0.15em\TeX}
\newcommand{\XYpic}{$\mathrm{X\kern-0.3em\raisebox{-0.18em}{Y}}$-$\mathrm{pic}\,$}

\newcommand{\cl}{C \kern -0.1em \ell}  

\newcommand{\BC}{\mathbb{C}}

\newcommand{\ed}{\end{document}}

\def\BC{{\mathbb C}}

\def\H{{\rm H}}

\def\mod{{\rm \;mod\; }}

\def\Mat{{\rm Mat}}

\def\Det{{\rm Det}}

\def\tr{{\rm tr}}

\newcommand{\C}{C \kern -0.1em \ell}

\begin{document}

%
%
%
%
%
%
%
%
%

\title[On Unitary Groups in Ternary and Generalized Clifford Algebras]
 {On Unitary Groups in Ternary and \\Generalized Clifford Algebras}
\author[D. Shirokov]{Dmitry Shirokov}
\address{%
HSE University\\
Moscow 101000\\
Russia
\medskip}
\address{
and
\medskip}
\address{
Institute for Information Transmission Problems of the Russian Academy of Sciences \\
Moscow 127051 \\
Russia}
\email{dm.shirokov@gmail.com}
\date{\today}
\dedicatory{Last Revised:\\ \today}
\begin{abstract}
We discuss a generalization of Clifford algebras known as generalized Clifford algebras (in particular, ternary Clifford algebras). In these objects, we have a fixed higher-degree form (in particular, a ternary form) instead of a quadratic form in ordinary Clifford algebras. We present a natural realization of unitary Lie groups, which are important in physics and other applications, using only operations in generalized Clifford algebras and without using the corresponding matrix representations. Basis-free definitions of the determinant, trace, and characteristic polynomial in generalized Clifford algebras are introduced. Explicit formulas for all coefficients of the characteristic polynomial and inverse in generalized Clifford algebras are presented. The operation of Hermitian conjugation (or Hermitian transpose) in generalized Clifford algebras is introduced without using the corresponding matrix representations. 
\end{abstract}
\label{page:firstblob}
\maketitle

\section{Introduction}\label{secIntro}

Over the past few years, a number of papers on ternary Clifford algebras with applications to physics and analysis have been published \cite{Paula, Paula2, Kerner, ShirT,Volk,Dubr,Masl}. Compared to ordinary (quadratic) Clifford algebras, instead of the quadratic form, a ternary form is fixed in ternary Clifford algebras. The formalism of generalized Clifford algebras $\cl^{\frac{1}{m}}_d$ (see \cite{Mor,Ram,Jag,Lin,Morris}), where we have a fixed $m$-ary form generalizes these concepts. From our point of view, ternary and generalized Clifford algebras are powerful mathematical tools that can be useful in algebra \cite{Paula, Mor, Morris, ShirT}, analysis \cite{Paula2}, physics \cite{Dubr, Jag, Kerner, Lin, Masl, Ram, Volk}, computer science, and engineering. See also the survey of applications \cite{Kwasn}.

The main topic of this paper is the study of unitary groups in the formalism of generalized Clifford algebras, which are important for all sorts of physical applications (for example, the group ${\rm SU}(3)$ is used to describe strong interactions in quantum chromodynamics). In this paper, we present an explicit realization of the unitary and special unitary Lie groups and the corresponding Lie algebras in the generalized Clifford algebras $\cl^{\frac{1}{m}}_d$. To do this, we introduce the concepts of Hermitian conjugation, determinant, and characteristic polynomial in $\cl^{\frac{1}{m}}_d$ without using the corresponding matrix representations. We present basis-free formulas for the inverse and all coefficients of the characteristic polynomial coefficients in the generalized Clifford algebras $\cl^{\frac{1}{m}}_d$. Theorems \ref{thHermG}, \ref{thUnitG}, \ref{FLG}, \ref{thunderG}, \ref{thSunitG}, Corollaries \ref{cor1}, \ref{cor2}, \ref{corGG}, \ref{corG}, \ref{cor3}, and Lemmas \ref{lem1}, \ref{lemG} are new.

The paper is organized as follows. In Section \ref{sec2}, we discuss the generalized Clifford algebras $\cl^{\frac{1}{m}}_d$ and their special case, ternary Clifford algebras $\cl^{\frac{1}{3}}_d$. In Section \ref{sec3}, we discuss matrix representations of $\cl^{\frac{1}{m}}_d$. In Section \ref{sec4}, we introduce the concepts of grades and projection operators in $\cl^{\frac{1}{m}}_d$. In Section \ref{sec5}, the operation of Hermitian conjugation (Hermitian transpose) in $\cl^{\frac{1}{m}}_d$ is introduced without using the corresponding matrix representations. In Section \ref{sec6}, we present an explicit realization of the unitary Lie groups and the corresponding Lie algebras in $\cl^{\frac{1}{m}}_d$ using the results of the previous sections of this paper. In Section \ref{sec7}, basis-free definitions of the determinant, trace, and characteristic polynomial in $\cl^{\frac{1}{m}}_d$ are introduced. Explicit formulas are presented for all coefficients of the characteristic polynomial and inverse in $\cl^{\frac{1}{m}}_d$. In Section \ref{sec8}, we present an explicit realization of the special unitary Lie groups and the corresponding Lie algebras in $\cl^{\frac{1}{m}}_d$. We present an explicit basis of the Lie algebra $\mathfrak{su}(3)$ in $\cl^{\frac{1}{3}}_{2}$ and its connection with the well-known Gell-Mann basis. The conclusions follow in Section \ref{sec9}. 

Note that the results on realization of the Lie group ${\rm SU}(3)$ and the corresponding Lie algebra $\mathfrak{su}(3)$ in the ternary Clifford algebra with two generators $\cl^{\frac{1}{3}}_{2}$ are presented in a short note (13 pages) in the Conference Proceedings \cite{su3}.
The current paper generalizes these results to the
case of unitary Lie groups of higher dimensions using the formalism of generalized Clifford algebras. Some results for the case $\cl^{\frac{1}{3}}_{2}$ are presented in this paper as examples to facilitate a better understanding of the more general case.

\section{Generalized Clifford algebras}
\label{sec2}

Let us consider the generalized (complex) Clifford algebra $\cl^{\frac{1}{m}}_d$ \cite{Mor,Ram,Jag,Lin} with $d$ generators that satisfy
\begin{eqnarray}
&&e_j^m=e,\qquad j=1, \ldots, d,\\
&&e_k e_l=\omega e_l e_k,\qquad k<l,\qquad \omega:=e^{\frac{2\pi i}{m}},
\end{eqnarray}
i.e. $\omega$ is a primitive $m$-th root of unity (the identity element) $e\equiv 1$ and $i$ is the imaginary unit.

We have $\dim(\cl^{\frac{1}{m}}_d)=m^d$. An arbitrary element $U\in\cl^{\frac{1}{m}}_d$ has the form
\begin{eqnarray}
U&=&\sum_{j_1, \ldots, j_d=0}^{m-1} u_{j_1 \ldots j_d} e_1^{j_1}\cdots e_d^{j_d}\label{UUG}\\
&=&u_{0\ldots 0}e+u_{10\ldots 0}e_1+ \cdots+u_{0\ldots 0 1} e_d+\cdots+u_{m-1\ldots m-1}e_1^{m-1} \cdots e_d^{m-1},\nonumber
\end{eqnarray}
where $u_{j_1 \ldots j_d}\in{\mathbb C}$.

In the case $m=2$, we get the standard (quadratic) Clifford algebra $\cl_d:=\cl^{\frac{1}{2}}_d$ with
\begin{eqnarray}
e_j^2=e,\qquad j=1, \ldots, d,\\
e_k e_l=-e_l e_k,\qquad k<l.
\end{eqnarray}
We have $\dim(\cl_d)=2^d$. An arbitrary element $U\in\cl_d$ has the form
\begin{eqnarray}
U&=&\sum_{j_1, \ldots, j_d=0}^{1} u_{j_1 \ldots j_d} e_1^{j_1}\cdots e_d^{j_d}\nonumber\\
&=&u_{0\ldots 0}e+u_{10\ldots 0}e_1+ \cdots+u_{0\ldots 0 1} e_d+\cdots+u_{1\ldots 1}e_1 \cdots e_d.
\end{eqnarray}

In the case $m=3$, we get the ternary Clifford algebra $\cl^{\frac{1}{3}}_d$ with
\begin{eqnarray}
&&e_j^3=e,\qquad j=1, \ldots, d,\\
&&e_k e_l=\omega e_l e_k,\qquad k<l,\qquad \omega:=e^{\frac{2\pi i}{3}}.
\end{eqnarray}
We have $\dim(\cl^{\frac{1}{3}}_d)=3^d$. An arbitrary element $U\in\cl^{\frac{1}{3}}_d$ has the form
\begin{eqnarray}
U&=&\sum_{j_1, \ldots, j_d=0}^{2} u_{j_1 \ldots j_d} e_1^{j_1}\cdots e_d^{j_d}\nonumber\\
&=&u_{0\ldots 0}e+u_{10\ldots 0}e_1+ \cdots+u_{0\ldots 0 1} e_d+\cdots+u_{2\ldots 2}e_1^{2} \cdots e_d^{2}.
\end{eqnarray}

\begin{ex}
Let us consider the particular case of ternary Clifford algebra with two generators 
$\cl^{\frac{1}{3}}_2$. The generators $e_1$ and $e_2$ satisfy
\begin{eqnarray}
&&e_1^3=e_2^3=e,\\
&&e_1 e_2=\omega e_2 e_1,\qquad \omega:=e^{\frac{2\pi i}{3}},
\end{eqnarray}
In this case, we have
$$
\omega^3=1,\qquad \omega^2+\omega+1=0,\qquad \overline{\omega}=\omega^2,\qquad \omega-\omega^2=i\sqrt{3}.
$$
An arbitrary element $U\in\cl^{\frac{1}{3}}_2$ has the form
\begin{eqnarray}
U=\sum_{j,k=0}^2 u_{jk} e_1^j e_2^k=u_{00}e+u_{10}e_1+u_{01}e_2+u_{20}e_1^2+u_{02}e_2^2+u_{11}e_1e_2\label{UU}\\
+u_{21}e_1^2 e_2+u_{12} e_1 e_2^2+u_{22}e_1^2 e_2^2,\qquad u_{jk}\in\mathbb{C}.
\end{eqnarray}
The multiplication table in $\cl^{\frac{1}{3}}_2$ is Table \ref{tab:multtable}.

\begin{table}[ht]
    \centering
    \!\begin{tabular}{c|cccccccc}
       1st $\backslash$ 2nd & $e_1$ & $e_2$ & $e_1^2$ & $e_2^2$ & $e_1 e_2$ & $e_1^2 e_2$ & $e_1 e_2^2$ & $e_1^2 e_2^2$ \\ \hline
       $e_1$  & $e_1^2$ & $e_1 e_2$ & $e$ & $e_1 e_2^2$ & $e_1^2 e_2$ & $e_2$ & $e_1^2 e_2^2$ & $e_2^2$\\
        $e_2$ & $\omega^2 e_1 e_2$ & $e_2^2$ & $\omega e_1^2 e_2$ & $e$ & $\omega^2 e_1 e_2^2$ & $\omega e_1^2 e_2^2$ & $\omega^2 e_1$ & $\omega e_1^2$\\
       $e_1^2$  & $e$ & $e_1^2 e_2$ & $e_1$ & $e_1^2 e_2^2$ & $e_2$ & $e_1 e_2$ & $e_2^2$ & $e_1 e_2^2$\\
       $e_2^2$  & $\omega e_1 e_2^2$ & $e$ & $\omega^2 e_1^2 e_2^2$ & $e_2$ & $\omega e_1$ & $\omega^2 e_1^2$ & $\omega e_1 e_2$ & $\omega^2 e_1^2 e_2$\\
       $e_1 e_2$  & $\omega^2 e_1^2 e_2$ & $e_1 e_2^2$ & $\omega e_2$ & $e_1$ & $\omega^2 e_1^2 e_2^2$ & $\omega e_2^2$ & $\omega^2 e_1^2$ & $\omega e$\\
       $e_1^2 e_2$  & $\omega^2 e_2$ & $e_1^2 e_2^2$ & $\omega e_1 e_2$ & $e_1^2$ & $\omega^2 e_2^2$ & $\omega e_1 e_2^2$ & $\omega^2 e$ & $\omega e_1$\\
        $e_1 e_2^2$ & $\omega e_1^2 e_2^2$ & $e_1$ & $\omega^2 e_2^2$ & $e_1 e_2$ & $\omega e_1^2$ & $\omega^2 e$ & $\omega e_1^2 e_2$ & $\omega^2 e_2$\\
        $e_1^2 e_2^2$ & $\omega e_2^2$ & $e_1^2$ & $\omega^2 e_1 e_2^2$ & $e_1^2 e_2$ & $\omega e$ & $\omega^2 e_1$ & $\omega e_2$ & $\omega^2 e_1 e_2$\\
    \end{tabular}
    \caption{Multiplication table in $\cl^{\frac{1}{3}}_2$.}
    \label{tab:multtable}
\end{table}

\end{ex}

\section{Matrix representations}\label{secmatrG}
\label{sec3}

We have the following isomorphisms (faithful representations) of the generalized Clifford algebras $\cl_d^{\frac{1}{m}}$ and the complex matrix algebras \cite{Morris}
\begin{eqnarray}
\!\!\!\!\!\!\!\!\!\!\!\!\!\!\!\!&&\gamma:\cl_d^{\frac{1}{m}}\to \label{gammaG}\\
\!\!\!\!\!\!\!\!\!\!\!\!\!\!\!\!&&\to\begin{cases}
\Mat(m^{\frac{d}{2}}, \BC),& \mbox{if $d$ is even;}\\
\underbrace{\Mat(m^{\frac{d-1}{2}}, \BC)\oplus\Mat(m^{\frac{d-1}{2}}, \BC)\oplus\cdots\oplus\Mat(m^{\frac{d-1}{2}}, \BC)}_{m\,\,\text{times}},& \mbox{if $d$ is odd.}
\end{cases}\nonumber
\end{eqnarray}
We denote the size of the corresponding matrices by
\begin{eqnarray}
N:=m^{[\frac{d+1}{2}]}=
\begin{cases}
m^{\frac{d}{2}},& \mbox{if $d$ is even};\\
m^{\frac{d+1}{2}},& \mbox{if $d$ is odd,}
\end{cases}
\end{eqnarray}
where $[\frac{d+1}{2}]$ is the integer part of $\frac{d+1}{2}$.

Let us fix one of such representations $\gamma$; we denote it by $\beta$. For the identity element $e$, we use the identity matrix $\beta(e)=I_{N}$ of the corresponding size. For the product of generators, we always use the matrix $\beta(e_{1}^{j_1}\cdots e_{d}^{j_d})=\beta(e_1)^{j_1}\cdots \beta(e_d)^{j_d}$. 
Thus, we only need to fix the matrices $\beta(e_i)$, $i=1, \ldots, d$.

Let us fix the following matrices  of size $m$ (see \cite{Morris})
\begin{eqnarray}
P:=\begin{bmatrix}
        0 & 1 & 0 & \ldots & 0\\
        0 & 0 & 1 & \ldots & 0\\
        \ldots & \ldots & \ldots & \ldots & \ldots \\
        0 & 0 & 0 & \ldots & 1\\
        1 & 0 & 0 & \ldots & 0
    \end{bmatrix};\quad 
Q:=\begin{bmatrix}
        0 & \omega & 0 & \ldots & 0\\
        0 & 0 & \omega^2 & \ldots & 0\\
        \ldots & \ldots & \ldots & \ldots & \ldots \\
        0 & 0 & 0 & \ldots & \omega^{m-1}\\
        1 & 0 & 0 & \ldots & 0
    \end{bmatrix}\label{PP}
\end{eqnarray}
for odd $m$ and
\begin{eqnarray}
Q:=\begin{bmatrix}
        0 & \zeta & 0 & \ldots & 0\\
        0 & 0 & \zeta^3 & \ldots & 0\\
        \ldots & \ldots & \ldots & \ldots & \ldots \\
        0 & 0 & 0 & \ldots & \zeta^{2m-3}\\
        \zeta^{2m-1} & 0 & 0 & \ldots & 0
    \end{bmatrix}\label{QQ}
\end{eqnarray}
for even $m$, where $\zeta$ is a primitive $2m$-th root of unity such that $\zeta^2=\omega$. Note that the matrices $P$ and $Q$ satisfy
\begin{eqnarray}
P^m=I_m,\qquad Q^m=I_m,\qquad PQ=\omega QP.
\end{eqnarray}
Let
\begin{eqnarray}
R:=\begin{cases}
P^{m-1}Q, & \mbox{if $m$ is odd;}\\
\zeta P^{m-1} Q, & \mbox{if $m$ is even},
\end{cases}\label{RR}
\end{eqnarray}
which is equal to the diagonal matrix ${\rm diag}(1, \omega, \omega^2, \ldots, \omega^{m-1})$ in both cases. Note that the matrices $P$, $Q$, and $R$ satisfy
\begin{eqnarray}
R^m=I_m,\qquad PR=\omega RP,\qquad QR=\omega RQ.
\end{eqnarray}
The matrices $P$ and $R$ are called {\it shift and clock matrices} in the literature, respectively.

Let us consider the algebra $\cl_d^{\frac{1}{m}}$ with even $d$. Then we can take
\begin{eqnarray}
\beta_i:=\beta(e_i)=\begin{cases}
\underbrace{R\otimes R \otimes \cdots \otimes R}_{i-1} \otimes P \otimes \underbrace{I \otimes \cdots \otimes I}_{\frac{d}{2}-i}, & i=1, 2, \ldots, \frac{d}{2},\\
\underbrace{R\otimes R \otimes \cdots \otimes R}_{i-\frac{d}{2}-1} \otimes Q \otimes \underbrace{I \otimes \cdots \otimes I}_{d-i}, & i=\frac{d}{2}+1, \frac{d}{2}+2, \ldots, d,
\end{cases}\label{betaG}
\end{eqnarray}
where each product has $\frac{d}{2}$ factors.

For the algebra $\cl_{d+1}^{\frac{1}{m}}$, where $d+1$ is odd, we can take 
\begin{eqnarray}
\beta(e_i)={\rm diag}(\beta_i, \omega\beta_i,\ldots, \omega^{m-1}\beta_i),\qquad i=1, \ldots, d+1,\label{betaGG}
\end{eqnarray}
where $\beta_{i}$, $i=1, \ldots, d$ are defined above and $\beta_{d+1}:=\underbrace{R\otimes R \otimes \cdots \otimes R}_{\frac{d}{2}}$.

Here we use the fact that the sets
\begin{eqnarray}
\omega^j \beta_i,\quad i=1, 2, \ldots, d+1,\quad j=0, 1, \ldots, m-1
\end{eqnarray}
generate $m$ inequivalent representations of $\cl_{d+1}^{\frac{1}{m}}$ with odd $d+1$.

\begin{ex}
In the case of the ternary Clifford algebra $\cl^{\frac{1}{3}}_2$ with two generators, we have the isomorphism 
\begin{eqnarray}
\gamma: \cl^{\frac{1}{3}}_2 \to {\rm Mat}(3, \mathbb{C}).\label{gamma}
\end{eqnarray}
Let us fix one specific matrix representation $\beta$ of the type (\ref{gamma}). We can take the following identity, shift, and clock matrices respectively:
\begin{eqnarray}
\beta(e)=
    \left[\begin{tabular}{ccc}
        $1$ & $0$ & $0$\\
        $0$ & $1$ & $0$\\
        $0$ & $0$ & $1$\\
    \end{tabular}\right],\,\,
    \beta(e_1)=
    \left[\begin{tabular}{ccc}
        $0$ & $1$ & $0$\\
        $0$ & $0$ & $1$\\
        $1$ & $0$ & $0$\\
    \end{tabular}\right],\,\,
    \beta(e_2)=
    \left[\begin{tabular}{ccc}
        $1$ & $0$ & $0$\\
        $0$ & $\omega$ & $0$\\
        $0$ & $0$ & $\omega^2$\\
    \end{tabular}\right].\label{beta}
\end{eqnarray}
Using 
\begin{eqnarray}
\beta(UV)=\beta(U) \beta(V),\qquad \forall U, V\in \cl^{\frac{1}{3}}_2,
\end{eqnarray}
we get
\begin{eqnarray}
   &&\beta(e_1^2)=
    \left[\begin{tabular}{ccc}
        $0$ & $0$ & $1$\\
        $1$ & $0$ & $0$\\
        $0$ & $1$ & $0$\\
    \end{tabular}\right],\,\,
    \beta(e_2^2)=
    \left[\begin{tabular}{ccc}
        $1$ & $0$ & $0$\\
        $0$ & $\omega^2$ & $0$\\
        $0$ & $0$ & $\omega$\\
    \end{tabular}\right],\nonumber\\
    &&\beta(e_{1}e_{2})=
    \left[\begin{tabular}{ccc}
        $0$ & $\omega$ & $0$\\
        $0$ & $0$ & $\omega^2$\\
        $1$ & $0$ & $0$\\
    \end{tabular}\right],\,\,
    \beta(e_1^2e_2)=
    \left[\begin{tabular}{ccc}
        $0$ & $0$ & $\omega^2$\\
       $1$ & $0$ & $0$\\
        $0$ & $\omega$ & $0$\\
    \end{tabular}\right],\\
    &&\beta(e_1 e_2^2)=
    \left[\begin{tabular}{ccc}
        $0$ & $\omega^2$ & $0$\\
        $0$ & $0$ & $\omega$\\
        $1$ & $0$ & $0$\\
    \end{tabular}\right],\,\,
    \beta(e_1^2 e_2^2)=
    \left[\begin{tabular}{ccc}
        $0$ & $0$ & $\omega$\\
        $1$ & $0$ & $0$\\
        $0$ & $\omega^2$ & $0$\\
    \end{tabular}\right].\nonumber
\end{eqnarray}
For an arbitrary $U\in\cl^{\frac{1}{3}}_2$ in the form (\ref{UU}), we have
\begin{eqnarray}
\beta(U)=\left[\begin{tabular}{ccc}
        $u_{00}+u_{01}+u_{02}\quad$ & $u_{10}+\omega u_{11}+\omega^2 u_{12}\quad$& $u_{20}+\omega u_{22}+\omega^2 u_{21}$\\
        $u_{20}+u_{21}+u_{22}\quad$ & $u_{00}+\omega u_{01}+\omega^2 u_{02}\quad$ & $u_{10}+\omega u_{12}+ \omega^2 u_{11}$ \\
        $u_{10}+u_{11}+u_{12}\quad$ & $u_{20}+\omega u_{21}+\omega^2 u_{22}\quad$ & $u_{00}+\omega u_{02}+\omega^2 u_{01}$\\
    \end{tabular}\right].\label{fixed2}
\end{eqnarray}
\end{ex}

\section{Subspaces of different grades and projection operators}
\label{sec4}

The generalized Clifford algebra $\cl^{\frac{1}{m}}_d$ is a direct sum of the subspaces of grades $k=0, 1, 2, 3, \ldots, d(m-1)$:
\begin{eqnarray}
\cl^{\frac{1}{m}}_d=\bigoplus_{k=0}^{d(m-1)} \cl^{\frac{1}{m}, k}_d,
\end{eqnarray}
where $\cl^{\frac{1}{m}, k}_d$ is a linear span of the basis elements $e_1^{j_1} e_2^{j_2} \cdots e_d^{j_d}$ with $j_1+j_2+\cdots+j_d=k$.

Let us denote the projection onto the subspace $\cl^{\frac{1}{m}, k}_d$ of grade $k$ by $\langle U \rangle_k$. For example, we have
\begin{eqnarray}
\langle U \rangle_0= u_{00\ldots 0}e \equiv u_{00 \ldots 0} \in \BC,\qquad \forall U\in \cl^{\frac{1}{m}}_d.
\end{eqnarray}
The operations $\langle U \rangle_k$, $k=0, 1, 2, \ldots, d(m-1)$, are linear:
\begin{eqnarray}
\langle \lambda U+\mu V \rangle_k=\lambda \langle U\rangle_k+\mu \langle V \rangle_k,\qquad \forall U, V\in\cl^{\frac{1}{m}}_d,\qquad \forall \lambda, \mu\in\BC. \label{linearG}
\end{eqnarray}

\begin{ex}
The ternary Clifford algebra $\cl^{\frac{1}{3}}_2$ is a direct sum of the five subspaces of grades $l=0, 1, 2, 3, 4$:
\begin{eqnarray}
\cl^{\frac{1}{3}}_2=\cl^{\frac{1}{3}, 0}_2\oplus \cl^{\frac{1}{3}, 1}_2\oplus \cl^{\frac{1}{3}, 2}_2\oplus\cl^{\frac{1}{3}, 3}_2\oplus\cl^{\frac{1}{3}, 4}_2,
\end{eqnarray}
where $\cl^{\frac{1}{3}, l}_2$ is a linear span of the basis elements $e_1^j e_2^k$ with $j+k=l$.
\end{ex}

\section{Hermitian conjugation}\label{sechermG}
\label{sec5}

Let us consider the following operation in $\cl^\frac{1}{m}_d$:
\begin{eqnarray}
\overline{U}:=U|_{u_{j_1 j_2 \ldots j_d} \to \bar{u}_{j_1 j_2 \ldots j_d},\, e_1^{j_1}e_2^{j_2}\cdots e_d^{j_d} \to (e_1^{j_1}e_2^{j_2}\cdots e_d^{j_d})^{-1}},\qquad \forall U\in\cl^\frac{1}{m}_d.\label{herm1G}
\end{eqnarray}
For an arbitrary $U$ of the form (\ref{UUG}), we get
\begin{eqnarray}
\overline{U}=\sum_{j_1, \ldots, j_d=0}^{m-1} \bar{u}_{j_1 \ldots j_d} (e_1^{j_1}\cdots e_d^{j_d})^{-1}=\overline{u}_{0\ldots 0}e+\overline{u}_{10\ldots 0}e_1^{m-1}+ \cdots\nonumber\\
+\overline{u}_{0\ldots 0 1} e_d^{m-1}+\cdots+\overline{u}_{m-1\ldots m-1}\omega^{\frac{(m-1)(d-1)d}{2}} e_1 \cdots e_d.\label{z3}
\end{eqnarray}
\begin{ex} In the case $\cl^\frac{1}{3}_2$, for an arbitrary $U$ of the form (\ref{UU}), we have
\begin{eqnarray}
\overline{U}=\sum_{j,k=0}^2 \overline{u_{jk}} (e_1^{j}e_2^{k})^{-1}=\overline{u_{00}}e+\overline{u_{10}}e_1^2+\overline{u_{01}}e_2^2+\overline{u_{20}}e_1+\overline{u_{02}}e_2+\nonumber\\
+\overline{u_{11}}\omega^2 e_1^2e_2^2+\overline{u_{21}}\omega e_1 e_2^2+\overline{u_{12}} \omega e_1^2 e_2+\overline{u_{22}}\omega^2 e_1 e_2. \label{UUh}
\end{eqnarray}
\end{ex}

This operation has the properties
\begin{eqnarray}
&&\overline{\overline{U}}=U,\qquad \overline{U+V}=\overline{U}+\overline{V},\qquad \overline{\alpha U}=\overline{\alpha}\, \overline{U},\label{pr1G}\\
&&\overline{UV}=\overline{V}\,\,\overline{U},\qquad \forall U, V\in\cl^{\frac{1}{m}}_d,\qquad \forall \alpha\in\BC.\label{pr2G}
\end{eqnarray}
We call the operation (\ref{herm1G}) {\it Hermitian conjugation (or Hermitian transpose) in the generalized Clifford algebra $\cl^{\frac{1}{m}}_d$} based on the following theorem. The similar operation is considered in ordinary (quadratic) Clifford algebras in \cite{Herm}.
\begin{thm}\label{thHermG}
For the fixed matrix representation $\beta$ (\ref{betaG}), (\ref{betaGG}), we have 
\begin{eqnarray}
(\beta(U))^{\rm H}=\beta(\overline{U}),\qquad \forall U\in\cl^{\frac{1}{m}}_d,\label{thhermG}
\end{eqnarray}
where we denote the Hermitian transpose of a matrix by ${\rm H}$.
\end{thm}
\begin{proof}
We verify directly that the matrices $\beta(e_a)$, $a=1, \ldots, d$, are unitary:
\begin{eqnarray}(\beta(e_a))^{\rm H}\beta(e_a)=I_N,\qquad a=1, \ldots, d,\label{rty}
\end{eqnarray}
using the explicit formulas for the representation $\beta$ from Section \ref{secmatrG}. This follows from the fact that the matrices $P$, $Q$, and $R$ (\ref{PP}), (\ref{QQ}), (\ref{RR}) are unitary and the well-known properties of the Kronecker product $(A\otimes B)^\H=A^\H \otimes B^\H$, $(A\otimes B)(C\otimes D)=AC \otimes BD$.

The condition (\ref{rty}) is equivalent to \begin{eqnarray}(\beta(e_a))^{\rm H}=\beta(e_a^{-1}),\qquad a=1, \ldots, d.
\end{eqnarray}
Further, we have 
\begin{eqnarray}
&&\!\!\!\!\!\!\!\!\!\!\beta(\overline{U})=\beta (\!\!\!\!\!\sum_{j_1, \ldots, j_d=0}^{m-1}\!\!\!\!\! \bar{u}_{j_1 \ldots j_d} (e_1^{j_1}\cdots e_d^{j_d})^{-1})=\!\!\!\!\!\!\!\sum_{j_1, \ldots, j_d=0}^{m-1}\!\!\!\!\! \bar{u}_{j_1 \ldots j_d} \beta((e_d^{j_d})^{-1})\cdots \beta((e_1^{j_1})^{-1})\nonumber\\
&&\!\!\!\!\!\!\!\!\!\!=\!\!\!\!\!\!\!\sum_{j_1, \ldots, j_d=0}^{m-1}\!\!\!\!\! \bar{u}_{j_1 \ldots j_d} (\beta(e_d^{j_d}))^{\rm H}\cdots (\beta(e_1^{j_1}))^{\rm H}=(\beta(\!\!\!\!\!\sum_{j_1, \ldots, j_d=0}^{m-1}\!\!\!\!\! u_{j_1 \ldots j_d} e_1^{j_1}\cdots e_d^{j_d}))^{\rm H}=(\beta(U))^{\rm H}.\nonumber
\end{eqnarray}
This completes the proof.
\end{proof}
Note that property (\ref{thhermG}) is not valid for an arbitrary matrix representation $\gamma$ of the form (\ref{gammaG}). We have 
\begin{eqnarray}
(\gamma(U))^{\rm H}=\gamma(\overline{U}),\qquad \forall U\in\cl^{\frac{1}{m}}_d
\end{eqnarray}
only for matrix representations $\gamma$ of the form $\gamma=T^{-1}\beta T$, where $\beta$ is (\ref{betaG}) and $T\in\cl^{\frac{1}{m}}_d$ is unitary (i.e. $\overline{T}T=e$; see Section \ref{sec6}).

The operation (\ref{herm1G}) allows us to introduce the inner product
\begin{eqnarray}
U\cdot V:= \langle \overline{U} V \rangle_0=\sum_{j_1, j_2, \ldots j_d=0}^{m-1} \overline{u_{j_1 j_2 \ldots j_d}} v_{j_1 j_2 \ldots j_d},\qquad \forall U, V\in\cl^{\frac{1}{m}}_d,\label{innerG}
\end{eqnarray}
with the properties
\begin{eqnarray}
\!\!\!\!\!\!\!\!\!\!\!\!\!\!\!\!\!\!\!\!&&U \cdot (\lambda V+\mu W)=\lambda\, U\cdot V+\mu\, U\cdot W,\quad \forall U, V, W\in \cl^{\frac{1}{m}}_d,\quad \forall\lambda, \mu\in\BC;\\
\!\!\!\!\!\!\!\!\!\!\!\!\!\!\!\!\!\!\!\!&&U\cdot V= \overline{V \cdot U},\quad \forall U, V\in \cl^{\frac{1}{m}}_d;\\
\!\!\!\!\!\!\!\!\!\!\!\!\!\!\!\!\!\!\!\!&&U \cdot U \geq 0,\qquad  \forall U\in\cl^{\frac{1}{m}}_d;\qquad U\cdot U=0 \Leftrightarrow U=0.
\end{eqnarray}
The inner product (\ref{innerG}) allows us to introduce the norm
$$
||U||:=\sqrt{U\cdot U}=\sqrt{\langle \overline{U} U \rangle_0}=\sqrt{\sum_{j_1, j_2, \ldots, j_d=0}^{m-1} |u_{j_1 j_2 \ldots j_d}|^2},\qquad \forall U\in\cl^{\frac{1}{m}}_d.
$$
with the properties
\begin{eqnarray}
&&||U||\geq 0,\qquad \forall U\in\cl^{\frac{1}{m}}_d;\qquad ||U||=0 \Leftrightarrow U=0;\\
&&||\lambda U||=|\lambda| \, ||U||,\qquad \forall U\in\cl^{\frac{1}{m}}_d,\qquad \forall \lambda\in\BC;\\
&&||U+V||\leq ||U||+||V||,\qquad \forall U, V\in\cl^{\frac{1}{m}}_d.
\end{eqnarray}
These facts about the inner product and the norm are well-known. We present them here for the convenience of the reader.

\section{Unitary Lie groups and Lie algebras}\label{sec6}

Let us consider the following {\it unitary Lie group in $\cl^{\frac{1}{m}}_d$}
\begin{eqnarray}
{\rm U}\cl^\frac{1}{m}_d:=\{U\in\cl^\frac{1}{m}_d:\quad \overline{U} U=e\}.\label{UclG}
\end{eqnarray}

\begin{thm}\label{thUnitG} We have the following isomorphism of the Lie group (\ref{UclG}) and the classical matrix unitary Lie groups
\begin{eqnarray}
{\rm U}\cl^\frac{1}{m}_d\simeq
\begin{cases}
\rm{U}(m^{\frac{d}{2}}),& \mbox{if $d$ is even;}\\
\underbrace{\rm{U}(m^{\frac{d-1}{2}})\times \rm{U}(m^{\frac{d-1}{2}})\times \cdots \times \rm{U}(m^{\frac{d-1}{2}})}_{m\,\, \text{times}},& \mbox{if $d$ is odd}.
\end{cases}
\end{eqnarray}
\end{thm}
\begin{proof} Theorem follows from Theorem \ref{thHermG} and the isomorphism (\ref{gammaG}).
\end{proof}
\begin{cor}\label{cor1}
The Lie algebra
\begin{eqnarray}
\mathfrak{u}\cl^\frac{1}{m}_d:=\{U\in\cl^\frac{1}{m}_d:\quad \overline{U} =-U\}
\end{eqnarray}
corresponds to the Lie group (\ref{UclG}) and is isomorphic to the classical matrix unitary Lie algebra
\begin{eqnarray}
\mathfrak{u}\cl^\frac{1}{m}_d\simeq\begin{cases}
\mathfrak{u}(m^{\frac{d}{2}}),& \mbox{if $d$ is even;}\\
\underbrace{\mathfrak{u}(m^{\frac{d-1}{2}})\oplus \mathfrak{u}(m^{\frac{d-1}{2}})\oplus \cdots \oplus \mathfrak{u}(m^{\frac{d-1}{2}})}_{m\,\, \text{times}},& \mbox{if $d$ is odd}.
\end{cases}
\end{eqnarray}
\end{cor}
The considered Lie groups and Lie algebras have the following (real) dimensions:
\begin{eqnarray}
\dim({\rm U}\cl^\frac{1}{m}_d)=\dim(\mathfrak{u}\cl^\frac{1}{m}_d)=m^d.
\end{eqnarray}

Note that the explicit basis of $\mathfrak{u}\cl^\frac{1}{m}_d$ is easily written out for fixed $m$ and $d$ using the relation $\overline{U} =-U$ with substitution of (\ref{UUG}) and (\ref{z3}) into it. For example, the basis contains the elements 
\begin{eqnarray*}
ie,\quad e_1-e_1^{m-1},\quad i(e_1+e_1^{m-1}),\quad  e_2-e_2^{m-1},\quad i(e_2+e_2^{m-1}),\quad \ldots
\end{eqnarray*}
and others ($m^d$ elements in total). In the case $m=3$ and $d=2$, the basis is written out in the example below.

\begin{ex}
In the case of $\cl^{\frac{1}{3}}_2$, we have the Lie group
\begin{eqnarray}
{\rm U}\cl^\frac{1}{3}_2:=\{U\in\cl^\frac{1}{3}_2:\quad \overline{U} U=e\},\label{Ucl}
\end{eqnarray}
which is isomorphic to the classical matrix unitary Lie group
\begin{eqnarray}
{\rm U}\cl^\frac{1}{3}_2\simeq{\rm U}(3):=\{A\in\Mat(3,{\mathbb C}):\quad A^{\rm H} A=I\}.
\end{eqnarray}
The Lie algebra
\begin{eqnarray}
\mathfrak{u}\cl^\frac{1}{3}_2:=\{U\in\cl^\frac{1}{3}_2:\quad \overline{U} =-U\}
\end{eqnarray}
corresponds to the Lie group (\ref{Ucl}) and is isomorphic to the classical matrix unitary Lie algebra
\begin{eqnarray}
\mathfrak{u}\cl^\frac{1}{3}_2\simeq\mathfrak{u}(3):=\{A\in\Mat(3,{\mathbb C}):\quad A^{\rm H}=-A\}.
\end{eqnarray}
The considered Lie group and Lie algebra have the following (real) dimensions:
\begin{eqnarray}
\dim({\rm U}\cl^\frac{1}{3}_2)=\dim(\mathfrak{u}\cl^\frac{1}{3}_2)=9.
\end{eqnarray}
A basis of $\mathfrak{u}\cl^\frac{1}{3}_2$ is 
\begin{eqnarray}
&&\tau_0=ie,\qquad \tau_1=e_1-e_1^2,\qquad
\tau_2=i(e_1+e_1^2),\label{newbasis0}\\
&&\tau_3=e_2-e_2^2,\qquad \tau_4=i(e_2+e_2^2),\nonumber\\
&&\tau_5=e_1e_2-\omega^2e_1^2e_2^2,\qquad \tau_6=i(e_1e_2+\omega^2 e_1^2 e_2^2),\nonumber\\
&&\tau_7=e_1^2e_2-\omega e_1 e_2^2,\qquad \tau_8=i(e_1^2 e_2+\omega e_1 e_2^2). \nonumber
\end{eqnarray}
\end{ex}

\section{Trace, determinant, characteristic polynomial, and inverse}\label{sec7}

We have the following lemma. 

\begin{lem}\label{lem1} For an arbitrary $\gamma$ (\ref{gammaG}), we have
\begin{eqnarray}
\tr(\gamma(U))=m^{[\frac{d+1}{2}]} \langle U \rangle_0,\qquad \forall
U\in \cl^{\frac{1}{m}}_d.\label{trG}
\end{eqnarray}
\end{lem}
\begin{proof} For the fixed matrix representation $\beta$ (\ref{betaG}), (\ref{betaGG}), it is easily verified. 

In the case of even $d$, the algebra $\cl^\frac{1}{m}_d$ is a central simple algebra. By the Skolem–Noether theorem, every automorphism of a central simple algebra is inner, so any other faithful  matrix representation of $\cl^\frac{1}{m}_d$ of minimal dimension $\gamma$ has the form $\gamma(U)=T^{-1} \beta(U) T$. We get $\tr(\gamma(U))=\tr(\beta(U))$ and the lemma is proved in this case. 

In the case of odd $d$, we can have one of the relations (see Section \ref{secmatrG}):
\begin{eqnarray}
\gamma(e_a)=T^{-1}\omega^j \beta(e_a) T,\qquad j=0, 1, \ldots, m-1.\label{ty1}
\end{eqnarray}
Using the grade automorphism (see the details in \cite{ShirT})
\begin{eqnarray}
\widehat{U}:=\bigoplus_{k=0}^{m(d-1)}  \omega^k \langle U \rangle_{k},\label{gr}
\end{eqnarray}
this can be rewritten in the form  \begin{eqnarray}
\gamma(U)=T^{-1}\beta(\widehat{U}^{(j)}) T,\qquad j=0, 1, \ldots, m-1,\label{ty2}
\end{eqnarray}
where $\widehat{U}^{(j)}$ means that the grade automorphism (\ref{gr}) is taken $j$ times. We get \begin{eqnarray}
\tr(\gamma(U))=\tr(\beta(\widehat{U}^{(j)}))=N\langle \widehat{U}^{(j)}\rangle=N \langle U \rangle_0=\tr(\beta(U))
\end{eqnarray}
and the lemma is proved in this case. 
\end{proof}
\begin{cor}\label{cor2} We have
\begin{eqnarray}
\langle UV \rangle_0=\langle VU \rangle_0,\qquad \forall U, V\in\cl^\frac{1}{m}_d.\label{UVVUG}
\end{eqnarray}
\end{cor}
\begin{proof}
It follows from (\ref{trG}) and the well-known property of the trace of matrices $\tr(AB)=\tr(BA)$, $\forall A, B\in\Mat(N, \BC)$.
\end{proof}

Let us introduce the notion of {\it determinant  in $\cl^\frac{1}{m}_d$}:
\begin{eqnarray}
\Det(U):=\det(\gamma(U))\in\mathbb{C},\qquad U\in\cl^\frac{1}{m}_d\label{detG}
\end{eqnarray}
for an arbitrary faithful matrix representation $\gamma$ (\ref{gammaG}) of $\cl^\frac{1}{m}_d$ of minimal dimension.

\begin{lem}\label{lemG} The determinant (\ref{detG}) is well-defined, i.e. it does not depend on the matrix representation $\gamma$ (\ref{gammaG}).
\end{lem}
\begin{proof} 
In the case of even $d$, the algebra $\cl^\frac{1}{m}_d$ is a central simple algebra. By the Skolem–Noether theorem, any faithful matrix representation of $\cl^\frac{1}{m}_d$ of minimal dimension $\gamma$ has the form $\gamma(U)=T^{-1} \beta(U) T$. We get $\det(\gamma(U))=\det(\beta(U))$ and the lemma is proved in this case. 

In the case of odd $d$, we can have one of the relations (\ref{ty1}), which can be rewritten in the form (\ref{ty2}). 
We get \begin{eqnarray}
\det(\gamma(U))=\det(\beta(\widehat{U}^{(j)})).
\end{eqnarray}
Now let us prove that
\begin{eqnarray}
\det(\beta(\widehat{U}^{(j)}))=\det(\beta(U)).
\end{eqnarray}
Let us introduce the subspaces (see \cite{ShirT})
\begin{eqnarray}
\cl^{\frac{1}{m}, (k)}_d:=\bigoplus_{j=k \mod m}\cl_d^{\frac{1}{m}, j}.
\end{eqnarray}
We have
\begin{eqnarray}
\cl_d^{\frac{1}{m}}=\bigoplus_{j=0}^{m-1}\cl^{\frac{1}{m}, (j)}_d,\quad U=\bigoplus_{j=0}^{m-1} U_{(j)},\quad U\in\cl_d^{\frac{1}{m}},\quad U_{(j)}\in\cl^{\frac{1}{m}, (j)}_d.
\end{eqnarray}
For the representation $\beta$ (\ref{betaGG}), we get 
\begin{eqnarray*}
&&\beta(U_{(i)})={\rm diag}(A_i, \omega^i A_i, \omega^{2i} A_i, \ldots, \omega^{(m-1)i}A_i),\quad i=0, 1, \ldots, m-1,\\
&&\beta(\widehat{U_{(i)}}^{(j)})={\rm diag}(\omega^{ij} A_i, \omega^{i+ij} A_i, \omega^{2i+ij} A_i, \ldots, \omega^{(m-1)i+ij}A_i),
\end{eqnarray*}
for some matrices $A_i$, $i=0, 1, \ldots, m-1$ of size $m^{\frac{d-1}{2}}$; $j=0, 1, \ldots, m-1$. We get
\begin{eqnarray*}
&&\!\!\!\!\!\!\!\!\!\!\beta(U)={\rm diag}(\sum_{i=0}^{m-1} A_i, \sum_{i=0}^{m-1}\omega^i A_i, \sum_{i=0}^{m-1}\omega^{2i} A_i, \ldots, \sum_{i=0}^{m-1}\omega^{(m-1)i}A_i),\\
&&\!\!\!\!\!\!\!\!\!\!\beta(\widehat{U}^{(j)})={\rm diag}(\sum_{i=0}^{m-1}\omega^{ij} A_i, \sum_{i=0}^{m-1}\omega^{i+ij} A_i, \sum_{i=0}^{m-1}\omega^{2i+ij} A_i, \ldots, \sum_{i=0}^{m-1}\omega^{(m-1)i+ij}A_i).
\end{eqnarray*}
Finally, we obtain
\begin{eqnarray*}
\det(\beta(\widehat{U}^{(j)}))=\prod_{k=0}^{m-1}\det(\sum_{i=0}^{m-1}\omega^{(k+j)i} A_i)=\prod_{k=0}^{m-1}\det(\sum_{i=0}^{m-1}\omega^{ki} A_i)=\det(\beta(U)),
\end{eqnarray*}
and the lemma is proved in this case.
\end{proof}
Note that in the particular case of ordinary ($m=2$) Clifford algebras $\cl_{d}^{\frac{1}{2}}$, we have $\omega=-1$ and the proof of the case of odd $d$ in Lemma \ref{lemG} has a simpler form and is presented in \cite{Det} (we simply have two matrices $\beta(U)={\rm diag}(A_0+A_1, A_0-A_1)$ and $\beta(\widehat{U})={\rm diag}(A_0+A_1, A_0-A_1)$ with the same determinant). In the case of ternary ($m=3$) Clifford algebra $\cl_{d}^{\frac{1}{3}}$ with odd $d$, we have three matrices 
\begin{eqnarray*}
&&\beta(U)={\rm diag}(A_0+A_1+A_2,\,\,\, A_0+\omega A_1+\omega^2 A_2,\,\,\, A_0+\omega^2 A_1+\omega A_2),\\
&&\beta(\widehat{U})={\rm diag}(A_0+\omega A_1+\omega^2 A_2,\,\,\, A_0+\omega^2 A_1+\omega A_2,\,\,\, A_0+A_1+A_2),\\ &&\beta(\widehat{\widehat{U}})={\rm diag}(A_0+\omega^2 A_1+\omega A_2,\,\,\, A_0+A_1+A_2,\,\,\, A_0+\omega A_1+\omega^2 A_2)
\end{eqnarray*}
with the same determinant.

\begin{cor}\label{corGG} We have
\begin{eqnarray}
\Det(U)=\overline{\Det(\overline{U})},\qquad \forall U\in\cl^{\frac{1}{m}}_d.
\end{eqnarray}
\end{cor}
\begin{proof}
We have
$$
\overline{\Det(\overline{U})}=\overline{\det(\beta(\overline{U}))}=\overline{\det(\beta(U))^{\rm H}}=\det(\beta(U))=\Det(U).$$
\end{proof}

Let us introduce {\it a characteristic polynomial of $U\in\cl^\frac{1}{m}_d$}:
\begin{eqnarray}
\!\!\varphi_U(\lambda):=\Det(\lambda e - U)=\lambda^N-C_{(1)}\lambda^{N-1}-\cdots-C_{(N)}\in{\mathbb C},\qquad \lambda\in{\mathbb C},
\end{eqnarray}
where $C_{(k)}=C_{(k)}(U)\in\cl^{\frac{1}{m}, 0}_d\equiv\BC$, $k=1, 2, 3, \ldots, N$, are {\it characteristic polynomial coefficients}.

\begin{thm}\label{FLG}
For the characteristic polynomial coefficients $C_{(k)}=C_{(k)}(U)\in\cl^{\frac{1}{m}, 0}_d\equiv\BC$, $k=1, 2, 3, \ldots, N$, we have the following basis-free formulas:
\begin{eqnarray}
U_{(1)}:=U,\qquad U_{(k+1)}=U(U_{(k)}-C_{(k)}),\qquad C_{(k)}:=\frac{N}{k}\langle U_{(k)}\rangle_0,\label{KG}\\
\Det(U)=-U_{(N)}=-C_{(N)}=U(C_{(N-1)}-U_{(N-1)})\in\BC.\label{LFG}
\end{eqnarray}
\end{thm}
\begin{proof} The formulas (\ref{KG}), (\ref{LFG}) are a  modification of the Faddeev--LeVerrier algorithm \cite{FL} for the case of generalized Clifford algebras, where we use the relation  (\ref{trG}).
\end{proof}

\begin{cor}\label{corG}
If $\Det(U)\neq 0$ for some $U\in\cl^{\frac{1}{m}}_d$, then $U$ is invertible and we have an explicit formula for the inverse:
\begin{eqnarray}
U^{-1}=\frac{C_{(N-1)}-U_{(N-1)}}{\Det(U)}.\label{U-1G}
\end{eqnarray}
\end{cor}
Here
\begin{eqnarray}
{\rm Adj}(U):={\rm adj}(\gamma(U))=C_{(N-1)}-U_{(N-1)}\label{adjG}
\end{eqnarray}
is the adjugate of an arbitrary element $U\in\cl^{\frac{1}{m}}_d$. We have
$$\Det(U)=U \,{\rm Adj}(U)={\rm Adj}(U)\, U.$$

We can also introduce the operation $\underline{U}$ to simplify the expressions (\ref{KG}), (\ref{LFG}). Let us have
\begin{eqnarray}
\underline{U}:=\langle U \rangle_0-\sum_{k=1}^{m(d-1)}{\langle U \rangle_k}=2\langle U\rangle_0-U,\qquad \forall U\in\cl^{\frac{1}{m}}_d.\label{underlineG}
\end{eqnarray}
The similar operation is considered in the ordinary (quadratic) Clifford algebras in \cite{Det}.

Using (\ref{trG}) and (\ref{underlineG}), we get
\begin{eqnarray}
\langle U\rangle_0=\frac{1}{2}(U+\underline{U})=\frac{1}{N}(\tr(\gamma(U))\in\BC\label{UovUG}
\end{eqnarray}
for an arbitrary $\gamma$ (\ref{gammaG}).

\begin{thm} \label{thunderG} We have the properties:
\begin{eqnarray}
&&\underline{\underline{U}}=U,\qquad \underline{U+V}=\underline{U}+\underline{V},\qquad \underline{\alpha U}=\alpha \underline{U},\label{proper0G}\\
&&\underline{UV}U=U\underline{VU},\qquad \forall U, V\in \cl^{\frac{1}{m}}_d,\qquad \forall \alpha \in\BC,\label{properG}\\
&&\underline{U \underline{V}}=\underline{U}\, \underline{V} +\underline{U} V-\underline{UV},\qquad \forall U, V\in\cl^{\frac{1}{m}}_d.\label{underG}
\end{eqnarray}
\end{thm}
\begin{proof} The first three properties (\ref{proper0G}) are trivial; see also (\ref{linearG}).

Let us prove the property (\ref{properG}). Using (\ref{UVVUG}) and (\ref{UovUG}), we get
$$UV+\underline{UV}=VU+\underline{VU}\in{\mathbb C}.$$
Multiplying the left side on the right by $U$ and the right side on the left by $U$, we get $\underline{UV}U=U\underline{VU}$. 

Let us prove the property (\ref{underG}). Using (\ref{proper0G}) and (\ref{UovUG}), we get
\begin{eqnarray}
\underline{U \underline{V}}+\underline{UV}=\underline{U \underline{V}+UV}=\underline{U(\underline{V}+V)}=\underline{U}(\underline{V}+V)=\underline{U}\, \underline{V} +\underline{U} V.
\end{eqnarray}
The theorem is proved. 
\end{proof}
The formulas for the determinant and other characteristic coefficients (\ref{KG}), (\ref{LFG}) can be simplified using the relation (\ref{UovUG}) and the properties (\ref{proper0G}), (\ref{properG}), (\ref{underG}). Namely, the property (\ref{underG}) effectively allows to flatten nested operations of type (\ref{underlineG}). Substituting $U=V$ into the property (\ref{underG}), we get
\begin{eqnarray}
\underline{U \underline{U}}=(\underline{U})^2+\underline{U} U-\underline{U^2},\qquad \forall U\in\cl^{\frac{1}{m}}_d\label{under2}.
\end{eqnarray}

\begin{ex}
For an arbitrary element $U\in\cl^\frac{1}{3}_2$ (\ref{UU}), we get the following explicit formula for the determinant
\begin{eqnarray}
&&\Det(U)=u_{00}^3+u_{10}^3+u_{01}^3+u_{20}^3+u_{02}^3+u_{11}^3+u_{21}^3+u_{12}^3+u_{22}^3\label{det0}\\
&&-3(u_{00}u_{01}u_{02}
+u_{10}u_{11}u_{12}+u_{00}u_{10}u_{20}+u_{01}u_{11}u_{21}+u_{02}u_{12}u_{22}\nonumber\\
&&+u_{20}u_{21}u_{22})-3\omega(u_{01}u_{12}u_{20}+u_{02}u_{10}u_{21}+u_{00}u_{11}u_{22})\nonumber\\
&&-3\omega^2(u_{02}u_{11}u_{20}+u_{00}u_{12}u_{21}+u_{01}u_{10}u_{22})\in\mathbb{C}.\nonumber
\end{eqnarray}
Using our techniques, we obtain another explicit formula for the determinant instead of (\ref{det0}), which is basis-free. The characteristic polynomial of $U\in\cl^\frac{1}{3}_2$
\begin{eqnarray}
\varphi_U(\lambda):=\Det(\lambda e - U)=\lambda^3-C_{(1)}\lambda^2-C_{(2)}\lambda-C_{(3)}\in{\mathbb C},\qquad \lambda\in{\mathbb C},
\end{eqnarray}
has three coefficients $C_{(k)}=C_{(k)}(U)\in\cl^{\frac{1}{3}, 0}_2\equiv\BC$, $k=1, 2, 3$, with the following basis-free formulas: 
\begin{eqnarray}
&&\tr(U)=C_{(1)}=\frac{3}{2}(U+\underline{U})\in\BC,\\
&&C_{(2)}=-\frac{3}{8}(U^2+\underline{U^2}+3U\underline{U}+3\underline{U \underline{U}})\in\BC,\label{C2}\\
&&\Det(U)=C_{(3)}=\frac{1}{8}U(-U^2-3U\underline{U}+3\underline{U^2}+9\underline{U\underline{U}})\in\BC.\label{det2}
\end{eqnarray}
In particular, all basis elements of $\cl^\frac{1}{3}_2$ have a determinant equal to $1$:
\begin{eqnarray*}
&&\Det(e)=\Det(e_1)=\Det(e_2)=\Det(e_1^2)=\Det(e_2^2)=\Det(e_1e_2)\\
&&=\Det(e_1^2e_2)=\Det(e_1 e_2^2) =\Det(e_1^2 e_2^2)=1.
\end{eqnarray*}
If $\Det(U)\neq 0$ for some $U\in\cl^{\frac{1}{3}}_2$, then $U$ is invertible and we have an explicit formula for the inverse:
\begin{eqnarray}
U^{-1}=\frac{-U^2-3U\underline{U}+3\underline{U^2}+9\underline{U\underline{U}}}{8\Det(U)}.\label{U-1}
\end{eqnarray}
Note that
\begin{eqnarray}
{\rm Adj}(U):={\rm adj}(\gamma(U))=\frac{1}{8}(-U^2-3U\underline{U}+3\underline{U^2}+9\underline{U\underline{U}})\label{adj}
\end{eqnarray}
is the adjugate of an arbitrary element $U\in\cl^{\frac{1}{3}}_2$ and
$$
\Det(U)=U \,{\rm Adj}(U)={\rm Adj}(U)\, U.
$$
Using (\ref{under2}), the formulas (\ref{C2}), (\ref{det2}), (\ref{U-1}), (\ref{adj}) can be rewritten in the following equivalent form:
\begin{eqnarray}
&&C_{(2)}=-\frac{3}{8}(U^2-2\underline{U^2}+6U\underline{U}+3(\underline{U})^2)\in\BC,\label{C2q}\\
&&\Det(U)=C_{(3)}=\frac{1}{8}U(-U^2-6\underline{U^2}+6U\underline{U}+9(\underline{U})^2)\in\BC,\label{det2q}\\
&&U^{-1}=\frac{-U^2-6\underline{U^2}+6U\underline{U}+9(\underline{U})^2}{8\Det(U)},\label{U-1q}\\
&&{\rm Adj}(U)=\frac{1}{8}(-U^2-6\underline{U^2}+6U\underline{U}+9(\underline{U})^2)\label{adjq}.
\end{eqnarray}
\end{ex}
\section{Special unitary Lie groups and Lie algebras}\label{sec8}

Let us consider the special unitary group in $\cl^\frac{1}{m}_d$
\begin{eqnarray}
{\rm SU}\cl^\frac{1}{m}_d:=\{U\in\cl^\frac{1}{m}_d:\quad \overline{U} U=e,\quad \Det(U)=1\},\label{SUclG}
\end{eqnarray}
where $\Det$ is defined by (\ref{LFG}).
\begin{thm}\label{thSunitG} We have the following isomorphism of the Lie group (\ref{SUclG}) and the classical matrix unitary Lie groups in the case of even $d$:
\begin{eqnarray}
{\rm SU}\cl^\frac{1}{m}_d\simeq
\rm{SU}(m^{\frac{d}{2}}).
\end{eqnarray}
\end{thm}
\begin{proof} Theorem follows from Theorem \ref{thUnitG} and the results of Section \ref{sec7}.
\end{proof}
\begin{cor}\label{cor3}
The Lie algebra
\begin{eqnarray}
\mathfrak{su}\cl^\frac{1}{m}_d:=\{U\in\cl^\frac{1}{m}_d:\quad \overline{U} =-U,\quad \langle U \rangle_0=0\}
\end{eqnarray}
corresponds to the Lie group (\ref{SUclG}) and is isomorphic to the classical matrix unitary Lie algebra in the case of even $d$:
\begin{eqnarray}
\mathfrak{su}\cl^\frac{1}{m}_d\simeq
\mathfrak{su}(m^{\frac{d}{2}}).
\end{eqnarray}
\end{cor}
The considered Lie groups and Lie algebras have the following (real) dimensions:
\begin{eqnarray}
\dim({\rm SU}\cl^\frac{1}{m}_d)=\dim(\mathfrak{su}\cl^\frac{1}{m}_d)=m^d-1.
\end{eqnarray}
\begin{ex} In the case of $\cl^\frac{1}{3}_2$, we have the Lie group
\begin{eqnarray}
{\rm SU}\cl^\frac{1}{3}_2:=\{U\in\cl^\frac{1}{3}_2:\quad \overline{U} U=e,\quad \Det(U)=1\},\label{SUcl}
\end{eqnarray}
where $\Det$ is defined as 
\begin{eqnarray}
\Det(U):=\frac{1}{8}U(-U^2-3U\underline{U}+3\underline{U^2}+9\underline{U\underline{U}})\in\BC.
\end{eqnarray}
We have the isomorphism to the classical matrix Lie group
\begin{eqnarray}
{\rm SU}\cl^\frac{1}{3}_2\simeq{\rm SU}(3):=\{A\in\Mat(3,{\mathbb C}):\quad A^{\rm H} A=I,\quad \det(A)=1\}.
\end{eqnarray}
The Lie algebra
\begin{eqnarray}
\mathfrak{su}\cl^\frac{1}{3}_2:=\{U\in\cl^\frac{1}{3}_2:\quad \overline{U} =-U,\quad \langle U \rangle_0=0\}
\end{eqnarray}
corresponds to the Lie group (\ref{SUcl}) and is isomorphic to the classical matrix Lie algebra
\begin{eqnarray}
\mathfrak{su}\cl^\frac{1}{3}_2\simeq \mathfrak{su}(3):=\{A\in\Mat(3,{\mathbb C}):\quad A^{\rm H} =-A,\quad \tr(A)=0\}.
\end{eqnarray}
The considered Lie group and Lie algebra have the following (real) dimensions:
\begin{eqnarray}
\dim({\rm SU}\cl^\frac{1}{3}_2)=  \dim(\mathfrak{su}\cl^\frac{1}{3}_2)=8.
\end{eqnarray}
A basis of $\mathfrak{su}\cl^\frac{1}{3}_2$ is 
\begin{eqnarray}
&&\tau_1=e_1-e_1^2,\qquad
\tau_2=i(e_1+e_1^2),\label{newbasis}\\
&&\tau_3=e_2-e_2^2,\qquad \tau_4=i(e_2+e_2^2),\nonumber\\
&&\tau_5=e_1e_2-\omega^2e_1^2e_2^2,\qquad \tau_6=i(e_1e_2+\omega^2 e_1^2 e_2^2),\nonumber\\
&&\tau_7=e_1^2e_2-\omega e_1 e_2^2,\qquad \tau_8=i(e_1^2 e_2+\omega e_1 e_2^2). \nonumber
\end{eqnarray}
For the fixed matrix representation $\beta$ (\ref{beta}), we get
\begin{eqnarray}
&&\beta(\tau_1)=
    \left[\begin{tabular}{ccc}
        $0$ & $1$ &  $-1$\\
        $-1$ & $0$ & $1$\\
        $1$ & $-1$ & $0$\\
    \end{tabular}\right],\quad \beta(\tau_2)=
    \left[\begin{tabular}{ccc}
        $0$ & $i$ & $i$\\
        $i$ & $0$ & $i$\\
        $i$ & $i$ & $0$\\
    \end{tabular}\right],\nonumber\\
 &&\label{basis1}
    \beta(\tau_3)=
    \left[\begin{tabular}{ccc}
        $0$ & $0$ & $0$\\
        $0$ & $\sqrt{3} i$ & $0$\\
        $0$ & $0$ & $-\sqrt{3} i$\\
    \end{tabular}\right],\quad
    \beta(\tau_4)=
    \left[\begin{tabular}{ccc}
        $2i$ & $0$ & $0$\\
        $0$ & $-i$ & $0$\\
        $0$ & $0$ & $-i$\\
    \end{tabular}\right],\\
    &&\beta(\tau_5)=
    \left[\begin{tabular}{ccc}
        $0$ & $\omega$ & $-1$\\
        $-\omega^2$ & $0$ & $\omega^2$\\
        $1$ & $-\omega$ & $0$\\
    \end{tabular}\right],\quad
    \beta(\tau_6)=
     \left[\begin{tabular}{ccc}
        $0$ & $i\omega$ & $i$\\
        $i\omega^2$ & $0$ & $i\omega^2$\\
        $i$ & $i\omega$ & $0$\\
    \end{tabular}\right],\nonumber\\
    &&\beta(\tau_7)=
    \left[\begin{tabular}{ccc}
        $0$ & $-1$ & $\omega^2$\\
        $1$ & $0$ & $-\omega^2$\\
        $-\omega$ & $\omega$ & $0$\\
    \end{tabular}\right],\quad
    \beta(\tau_8)=
    \left[\begin{tabular}{ccc}
        $0$ & $i$ & $i\omega^2$\\
        $i$ & $0$ & $i\omega^2$\\
        $i\omega$ & $i\omega$ & $0$\\
    \end{tabular}\right].\nonumber
\end{eqnarray}
All these matrices are anti-Hermitian $(\beta(\tau_j))^{\rm H}=-\beta(\tau_j)$, $j=1, \ldots, 8$.

In physics, another basis of $\mathfrak{su}(3)$
\begin{eqnarray}
\theta_j=i \lambda_j,\qquad j=1, \ldots, 9\label{GMb}
\end{eqnarray}
constructed using the well-known Gell-Mann matrices \cite{GM}
\begin{eqnarray}
\!\!\!\!\!\!\!\!\!\!&&\lambda_1=\left[\begin{tabular}{ccc}
        $0$ & $1$ & $0$\\
        $1$ & $0$ & $0$\\
        $0$ & $0$ & $0$\\
    \end{tabular}\right],\quad 
    \lambda_2=\left[\begin{tabular}{ccc}
        $0$ & $-i$ & $0$\\
        $i$ & $0$ & $0$\\
        $0$ & $0$ & $0$\\
    \end{tabular}\right],\quad 
    \lambda_3=\left[\begin{tabular}{ccc}
        $1$ & $0$ & $0$\\
        $0$ & $-1$ & $0$\\
        $0$ & $0$ & $0$\\
    \end{tabular}\right],\nonumber\\
    \!\!\!\!\!\!\!\!\!\!&&\lambda_4=\left[\begin{tabular}{ccc}
        $0$ & $0$ & $1$\\
        $0$ & $0$ & $0$\\
        $1$ & $0$ & $0$\\
    \end{tabular}\right],\quad 
    \lambda_5=\left[\begin{tabular}{ccc}
        $0$ & $0$ & $-i$\\
        $0$ & $0$ & $0$\\
        $i$ & $0$ & $0$\\
    \end{tabular}\right],\quad 
    \lambda_6=\left[\begin{tabular}{ccc}
        $0$ & $0$ & $0$\\
        $0$ & $0$ & $1$\\
        $0$ & $1$ & $0$\\
    \end{tabular}\right],\label{GM}\\
    \!\!\!\!\!\!\!\!\!\!&&\lambda_7=\left[\begin{tabular}{ccc}
        $0$ & $0$ & $0$\\
        $0$ & $0$ & $-i$\\
        $0$ & $i$ & $0$\\
    \end{tabular}\right],\quad 
    \lambda_8=\frac{1}{\sqrt{3}}\left[\begin{tabular}{ccc}
        $1$ & $0$ & $0$\\
        $0$ & $1$ & $0$\\
        $0$ & $0$ & $-2$\\
    \end{tabular}\right].\nonumber
\end{eqnarray}
is usually used. We represent the explicit connection between the basis (\ref{basis1}) and the Gell-Mann basis (\ref{GMb}) below.

Note that
\begin{eqnarray}
(\lambda_j)^{\rm H}=\lambda_j,\qquad (\theta_j)^{\rm H}=-\theta_j,\qquad j=0, 1, \ldots, 8.
\end{eqnarray}
Let us use the notation $\beta_j:=\beta(\tau_j)$, $j=1, \ldots, 8$ for the basis (\ref{basis1}). The explicit relation between the bases (\ref{basis1}) and (\ref{GMb}) is the following:
\begin{eqnarray}
&&\beta_1=\theta_2-\theta_5+\theta_7,\qquad \beta_2=\theta_1+\theta_4+\theta_6,\\
&&\beta_3=-\frac{\sqrt{3}}{2}\theta_3+\frac{3}{2}\theta_8,\qquad \beta_4=\frac{3}{2}\theta_3+\frac{\sqrt{3}}{2}\theta_8,\\
&&\beta_5=\frac{\sqrt{3}}{2}\theta_1-\frac{1}{2}\theta_2-\theta_5,\qquad \beta_6=-\frac{1}{2}\theta_1-\frac{\sqrt{3}}{2}\theta_2+\theta_4,\\
&&\beta_7=-\theta_2-\frac{\sqrt{3}}{2}\theta_4-\frac{1}{2}\theta_5+\frac{\sqrt{3}}{2}\theta_6+\frac{1}{2}\theta_7,\\
&&\beta_8=\theta_1-\frac{1}{2}\theta_4+\frac{\sqrt{3}}{2}\theta_5-\frac{1}{6}\theta_6+\frac{\sqrt{3}}{2}\theta_7.
\end{eqnarray}
In the other way:
\begin{eqnarray}
&&\theta_1=-\frac{1}{4\sqrt{3}}\beta_1+\frac{1}{4}\beta_2+\frac{\sqrt{3}}{4}\beta_5-\frac{1}{4}\beta_6-\frac{1}{4\sqrt{3}}\beta_7+\frac{1}{4}\beta_8,\\
&&\theta_2=\frac{1}{4}\beta_1+\frac{1}{4\sqrt{3}}\beta_2-\frac{1}{4}\beta_5-\frac{\sqrt{3}}{4}\beta_6-\frac{1}{4}\beta_7-\frac{1}{4\sqrt{3}}\beta_8,\\
&&\theta_3=-\frac{1}{2\sqrt{3}}\beta_3+\frac{1}{2}\beta_4,\\
&&\theta_4=\frac{1}{4\sqrt{3}}\beta_1+\frac{1}{4}\beta_2+\frac{1}{2}\beta_6-\frac{1}{2\sqrt{3}}\beta_7,\\
&&\theta_5=-\frac{1}{4}\beta_1+\frac{1}{4\sqrt{3}}\beta_2-\frac{1}{2}\beta_5+\frac{1}{2\sqrt{3}}\beta_8,\\
&&\theta_6=\frac{1}{2}\beta_2-\frac{\sqrt{3}}{4}\beta_5-\frac{1}{4}\beta_6+\frac{\sqrt{3}}{4}\beta_7-\frac{1}{4}\beta_8,\\
&&\theta_7=\frac{1}{2}\beta_1-\frac{1}{4}\beta_5+\frac{\sqrt{3}}{4}\beta_6+\frac{1}{4}\beta_7+\frac{\sqrt{3}}{4}\beta_8,\\
&&\theta_8=\frac{1}{2}\beta_3+\frac{1}{2\sqrt{3}}\beta_4.
\end{eqnarray}
Note that when we use the matrix formalism, the Gell-Mann basis (\ref{GMb}) is more useful since all the matrices (\ref{GM}) have only two or three non-zero elements. If we use the formalism of ternary Clifford algebra, then the basis (\ref{newbasis}) is more useful since all the basis elements have only two non-zero coefficients.
\end{ex}

\section{Conclusions}\label{sec9}

In this paper, we present an explicit realization of the unitary and special unitary Lie groups and the corresponding Lie algebras in the generalized Clifford algebras $\cl^{\frac{1}{m}}_d$ (in particular, ternary Clifford algebras $\cl^{\frac{1}{3}}_d$) using the basis-free concepts of the determinant and Hermitian conjugation in $\cl^{\frac{1}{m}}_d$. This realization uses only operations in the generalized Clifford algebras $\cl^{\frac{1}{m}}_d$ and does not use the corresponding matrix representations. We introduce the operation of Hermitian conjugation, the characteristic polynomial, the trace, and the determinant in an arbitrary generalized Clifford algebra. We present basis-free formulas for the inverse and all coefficients of the characteristic polynomial coefficients in the generalized Clifford algebras $\cl^{\frac{1}{m}}_d$. The unitary and special unitary Lie groups and the corresponding Lie algebras are widely used in physics. We expect further development of the theory of ternary and generalized Clifford algebras and their applications in physics, engineering, and computer science.

\section*{Acknowledgements}

The results of this paper were reported at the ENGAGE Workshop (Geneva, Switzerland, July 2024) within the International Conference Computer Graphics International 2024 (CGI 2024). The author is grateful to the organizers and the participants of this conference for fruitful discussions.

The author is grateful to the anonymous reviewers for their careful reading of the paper and helpful comments on how to improve the presentation.

The article was prepared within the framework of the project “Mirror Laboratories” HSE University “Quaternions, geometric algebras and applications”.

\medskip

\noindent{\bf Data availability} Data sharing not applicable to this article as no datasets were generated or analyzed during the current study.

\medskip

\noindent{\bf Declarations}\\
\noindent{\bf Conflict of interest} The authors declare that they have no conflict of interest.
  
\bibliographystyle{spmpsci}

\end{document}